\documentclass[3p,12pt]{elsarticle}

\usepackage{amssymb,latexsym,amsmath,color,amsthm,url,multirow}

\journal{}

\newcommand{\eps}{\varepsilon}
\newcommand{\set}[1]{\left\{#1\right\}}
\newcommand{\abs}[1]{\left|#1\right|}
\newcommand{\p}{\partial}
\newcommand{\mE}{\mathbf{E}}
\newcommand{\mU}{\mathbf{U}}
\newcommand{\mV}{\mathbf{V}}

\newcommand{\mf}{\mathbf{f}}
\newcommand{\mg}{\mathbf{g}}

\newcommand{\mx}{\mathbf{x}}
\newcommand{\my}{\mathbf{y}}
\newcommand{\mz}{\mathbf{z}}
\newcommand{\vn}{\boldsymbol{\nu}}
\newcommand{\vp}{\boldsymbol{\phi}}
\newcommand{\vt}{\boldsymbol{\theta}}
\newcommand{\vx}{\boldsymbol{\xi}}

\theoremstyle{plain}
\newtheorem{thm}{Theorem}[section]

\newtheorem{lem}[thm]{Lemma}

\newtheorem{Rem}[thm]{Remark}

\begin{document}

\begin{frontmatter}



\title{Structure and properties of linear sampling method for perfectly conducting, arc-like cracks}


\author{Won-Kwang Park}
\ead{parkwk@kookmin.ac.kr}
\address{Department of Mathematics, Kookmin University, Seoul, 136-702, Korea.}

\begin{abstract}
  We consider the imaging of arbitrary shaped, arc-like perfectly conducting cracks located in the two-dimensional homogeneous space via linear sampling method. Based on the structure of eigenvectors of so-called Multi Static Response (MSR) matrix, we discover the relationship between imaging functional adopted in the linear sampling method and Bessel function of integer order of the first kind. This relationship tells us that why linear sampling method works for imaging of perfectly conducting cracks in either Transverse Magnetic (Dirichlet boundary condition) and Transverse Electric (Neumann boundary condition), and explains its certain properties. Furthermore, we suggest multi-frequency imaging functional, which improves traditional linear sampling method. Various numerical experiments are performed for supporting our explores.
\end{abstract}


\end{frontmatter}





\section{Introduction}
Linear sampling methods for inverse scattering from one or several unknown scatterers via the far-field patterns of scattered fields were originally introduced by Colton and Kirsch \cite{CK}. It has been applied very successfully in various inverse scattering problems. Related works can be found in \cite{CC,CGK,C,CHM,HM,KR,NG} and references therein. Nowadays, linear sampling method strongly contributes in inverse problems as they have already made earlier in many others. This is due to the fact that it is fast to perform, is robust with respect to the unavoidable noise, requires no \textit{a priori} information about the scatterers, and can be applied to the various types of scatterers such as point-like scatterers, cracks, and arbitrary shaped extended inhomogeneities.

Despite of its successful applications in various area, the use of linear sampling method were heuristic without rigorous justification. Due to this reason, many authors tried to confirm why this method works, see \cite{A,AL,H}. Although, these remarkable features have shown feasibilities as a non-iterative imaging technique, a detailed structure analysis of imaging functional in linear sampling method must be considered since some phenomena cannot be explained in traditional approach (such as unexpected appearance of ghost replicas, etc.). The above phenomena gives us an motivation of exploring the structure of imaging functional and correspondingly discover certain properties of linear sampling method in the reconstruction of arbitrary shaped perfectly conducting, arc-like cracks. Our exploration is based on the physical factorization and structure of singular and eigenvectors of so called Multi-Static Response (MSR) matrix \cite{HSZ}. With this, we will establish a relationship between imaging functional in linear sampling method and Bessel function of integer order of the first kind. Furthermore, based on this structure, we introduce a multi-frequency based linear sampling method in order to improve traditional one.

This paper is organized as follows. In section \ref{sec:2}, we briefly introduce two-dimensional direct scattering problems in the existence of perfectly conducting cracks and linear sampling method. In section \ref{sec:3}, the structure of imaging functional in linear sampling method is explored, its certain properties are discussed, and improved multi-frequency based imaging functional is considered. In section \ref{sec:4}, various numerical examples are exhibited for supporting explored structure and properties. A short conclusion is mentioned in section \ref{sec:5}.

\section{Direct scattering problem and linear sampling method}\label{sec:2}
In this section, we briefly introduce the basic concept of direct scattering problem in the present of perfectly conducting crack and linear sampling method. A more detailed description can be found in various works \cite{AK,CC,C,CHM,KR}.

\subsection{Direct scattering problem}
Let $\Gamma$ denotes the arbitrary shaped perfectly conducting crack located in the homogeneous space $\mathbb{R}^2$. This is an oriented piecewise smooth nonintersecting arc without cusp that can be represented as
\begin{equation}\label{crack}
\Gamma=\set{\vp(s):s\in[a,b]}
\end{equation}
where $\vp:[a,b]\longrightarrow\mathbb{R}^2$ is an injective piecewise $\mathcal{C}^3$ function.

In the existence of $\Gamma$, let $u(\mx,\vt;k)$ be the time-harmonic total field that satisfies the two-dimensional Helmholtz wave equation
\begin{equation}\label{HelmEQ}
\Delta u(\mx,\vt;k)+k^2 u(\mx,\vt;k)=0\quad\mbox{in}\quad\mathbb{R}^2\backslash\Gamma
\end{equation}
with incident direction $\vt$ and strictly positive wave number $k=\omega\sqrt{\mu\eps}$, letting $\eps$ be the electric permittivity and $\mu$ the magnetic permeability. Throughout this paper, applied wave number $k$ is of the form $k=2\pi/\lambda$, where $\lambda$ denotes the wavelength. For the case of Transverse Magnetic (TM) polarization, the field cannot penetrate into $\Gamma$, i.e., $u(\mx,\vt;k)$ satisfies the Dirichlet boundary condition
\begin{equation}\label{HelmBC}
u(\mx,\vt;k)=0\quad\mbox{on}\quad\Gamma.
\end{equation}
Conversely, let us consider the Transverse Electric (TE) polarization case, letting $u(\mx,\vt;k)$ be the (single-component) magnetic field that satisfies the two-dimensional Helmholtz wave equation (\ref{HelmEQ}) with the Neumann boundary condition on $\Gamma$:
\begin{equation}\label{HelmNC}
\frac{\p u(\mx,\vt;k)}{\p\vn(\mx)}=0\quad\mbox{on}\quad\Gamma\backslash\set{\vp(a),\vp(b)},
\end{equation}
where $\vn(\mx)$ is a unit normal vector to $\Gamma$ at $\mx$.

Let us notice that the total field can always be decomposed as $u(\mx,\vt;k)=u_{\mathrm{inc}}(\mx,\vt;k)+u_{\mathrm{scat}}(\mx,\vt;k)$, where $u_{\mathrm{inc}}(\mx,\vt;k)=\exp(ik\vt\cdot\mx)$ is the given incident field, and $u_{\mathrm{scat}}(\mx,\vt;k)$ is the unknown scattered field, which is required to satisfy the Sommerfeld radiation condition
\[\lim_{\abs{\mx}\to\infty}\sqrt{\abs{\mx}}\left(\frac{\p u_{\mathrm{scat}}(\mx,\vt;k)}{\p\abs{\mx}}-iku_{\mathrm{scat}}(\mx,\vt;k)\right)=0\]
uniformly into all directions $\hat{\mx}=\frac{\mx}{\abs{\mx}}$. For TM case, by \cite{K}, $u_{\mathrm{scat}}(\mx,\vt;k)$ can be represented by the form of a single-layer potential
\begin{equation}\label{SLP}
  u_{\mathrm{scat}}(\mx,\vt;k)=\int_{\Gamma}\Phi(\mx,\my;k)\varphi(\my,\vt;k)d\my\quad\mbox{for}\quad\mx\in\mathbb{R}^2\backslash\Gamma
\end{equation}
with unknown density function $\varphi(\my,\vt;k)$. Here, $\Phi(\mx,\my;k)$ the two-dimensional fundamental solution to the Helmholtz equation
\[\Phi(\mx,\my;k)=\frac{i}{4}H_0^1(k\abs{\mx-\my})\quad\mbox{for}\quad\mx\ne\my,\]
expressed in terms of the Hankel function $H_0^1$ of order zero and of the first kind. And for TE case, by \cite{M1}, $u_{\mathrm{scat}}(\mx,\vt;k)$ can be represented by the form of a double-layer potential
\begin{equation}\label{DLP}
  u_{\mathrm{scat}}(\mx,\vt;k)=\int_{\Gamma}\frac{\p\Phi(\mx,\my;k)}{\p\vn(\my)}\psi(\my,\vt;k)d\my\quad\mbox{for}\quad \mx\in\mathbb{R}^2\backslash\Gamma.
\end{equation}

The far-field pattern $u_{\infty}(\hat{\mx},\vt;k)$ of the scattered field $u_{\mathrm{scat}}(\mx,\vt;k)$ is defined on $\mathbb{S}^1$. It can be represented as
\[u_{\mathrm{scat}}(\mx,\vt;k)=\frac{\exp(ik\abs{\mx})}{\sqrt{\abs{\mx}}}\bigg(u_{\infty}(\hat{\mx},\vt;k)+\mathcal{O}\left(\frac{1}{\abs{\mx}}\right)\bigg)\]
uniformly in all directions $\hat{\mx}=\mx/\abs{\mx}$ and $\abs{\mx}\longrightarrow\infty$. For TM case, the far field pattern is represented as follows
\begin{align}
\begin{aligned}\label{FFPD}
u_\infty(\hat{\mx},\vt;k)&=-\frac{\exp(i\frac{\pi}{4})}{\sqrt{8\pi k}}\int_{\Gamma}e^{-ik\hat{\mx}\cdot\my}\left(\frac{\p u_{+}(\my,\vt;k)}{\p\vn(\my)}-\frac{\p u_{-}(\my,\vt;k)}{\p\vn(\my)}\right)d\my\\
&=\frac{1+i}{4\sqrt{\pi k}}\int_{\Gamma}e^{-ik\hat{\mx}\cdot\my}\varphi(\my,\vt;k)d\my.
\end{aligned}
\end{align}

Similarly, the far-field pattern for TE case is given by
\begin{align}
\begin{aligned}\label{FFPN}
u_\infty(\hat{\mx},\vt;k)&=-\frac{\exp(i\frac{\pi}{4})}{\sqrt{8\pi k}}\int_{\Gamma}\frac{\p e^{-ik\hat{\mx}\cdot\my}}{\p\vn(\my)}\bigg(u_{+}(\my,\vt;k)-u_{-}(\my,\vt;k)\bigg)d\my\\
&=\frac{1-i}{4}\sqrt{\frac{k}{\pi}}\int_{\Gamma}\hat{\mx}\cdot\vn(\my)e^{-ik\hat{\mx}\cdot\my}\psi(\my,\vt;k)d\my.
\end{aligned}
\end{align}

\subsection{Linear sampling method}
In this section, we apply the far-field pattern formulas (\ref{FFPD}) and (\ref{FFPN}) in order to introduce the linear sampling method. For this purpose, we consider the following Multi-Static Response (MSR) matrix
\[\mathbb{K}(k):=\bigg[K_{jl}(\hat{\mx}_j,\vt_l;k)\bigg]_{j,l=1}^{N}=\bigg[u_\infty(\hat{\mx}_j,\vt_l;k)\bigg]_{j,l=1}^{N}.\]
From now on, we assume that the directions of incident and observation are coincide i.e., if $\hat{\mx}_j=-\vt_j$. Then, for TM case, the MSR matrix $\mathbb{K}$ can be written as
\begin{equation}\label{MSRD}
\mathbb{K}(k)=\frac{1+i}{4\sqrt{\pi k}}\int_{\Gamma}\mathbb{E}_{\mathrm{D}}(\hat{\mx},\my;k)\mathbb{F}_{\mathrm{D}}(\hat{\mx},\my;k)^Td\my,
\end{equation}
where $\mathbb{E}_{\mathrm{D}}(\hat{\mx},\my;k)$ is the illumination vector
\begin{align}
\begin{aligned}\label{VecED}
\mathbb{E}_{\mathrm{D}}(\hat{\mx},\my;k)&=\bigg[\exp(-ik\hat{\mx}_1\cdot\my),\exp(-ik\hat{\mx}_2\cdot\my),\cdots,\exp(-ik\hat{\mx}_N\cdot\my)\bigg]^T\bigg|_{\hat{\mx}_j=-\vt_j}\\ &=\bigg[\exp(ik\vt_1\cdot\my),\exp(ik\vt_2\cdot\my),\cdots,\exp(ik\vt_N\cdot\my)\bigg]^T
\end{aligned}
\end{align}
and where $\mathbb{F}_{\mathrm{D}}(\hat{\mx},\my;k)$ is the resulting density vector
\begin{equation}
\mathbb{F}_{\mathrm{D}}(\hat{\mx},\my;k)=\bigg[\varphi(\my,\vt_1;k),\varphi(\my,\vt_2;k),\cdots,\varphi(\my,\vt_N;k)\bigg]^T.
\end{equation}
Here, $\set{\hat{\mx}_j}_{j=1}^N\subset\mathbb{S}^1$ is a discrete finite set of observation directions and $\set{\vt_l}_{l=1}^N\subset\mathbb{S}^1$ is the same number of incident directions.

Formula (\ref{MSRD}) is a factorization of the MSR matrix that separates the known incoming wave information from the unknown information. The range of $\mathbb{K}(k)$ is determined by the span of the $\mathbb{E}_{\mathrm{D}}(\hat{\mx},\my;k)$ corresponding to the $\Gamma$, i.e., we can define a signal subspace by using a set of left singular vectors of $\mathbb{K}(k)$. We refer to \cite{HSZ,PL1} for a detailed discussion.

Assume that the crack is divided into $M$ different segments of size of order half the wavelength $\lambda/2$. Having in mind the Rayleigh resolution limit, any detail less than one-half of the wavelength cannot be probed, and only one point, say $\my_m$ for $m=1,2,\cdots,M$, at each segment is expected to contribute at the image space of the response matrix $\mathbb{K}(k)$, refer to \cite{AKLP,PL1,PL3}.

Since the coincide configuration of incident and observation directions, MSR matrix $\mathbb{K}(k)$ is complex symmetric but not Hermitian, refer to \cite{C,HHSZ,P3,P4,PL1,PL2,PL3}. Hence, one must consider the Hermitian matrix $\mathbb{A}(k)=\mathbb{K}(k)^*\mathbb{K}(k)$ with eigenvalues $\set{\sigma_1(k),\sigma_2(k),\cdots,\sigma_N(k)}$ satisfying $\sigma_1(k)\geq\sigma_2(k)\geq\cdots\geq\sigma_N(k)$ and corresponding eigenvectors $\set{\mE_1(k),\mE_2(k),\cdots,\mE_N(k)}$. Then, the range of $\mathbb{A}^{1/4}$ can be determined from the eigenvalues and eigenvectors of $\mathbb{A}$ such that
\[\mathrm{Ran}\mathbb{A}^{1/4}=\set{\mf:\sum_{n=1}^{N}\frac{|\langle\mE_n(k),\mf\rangle|^2}{\sqrt{|\sigma_n(k)|}}<+\infty}.\]
Now, by introducing a test vector $\mg_{\mathrm{D}}(\mz;k)$ as
\begin{equation}\label{testvector}
  \mg_{\mathrm{D}}(\mz;k):=\frac{1}{\sqrt{N}}\bigg[\exp(ik\vt_1\cdot\mz),\exp(ik\vt_2\cdot\mz),\cdots,\exp(ik\vt_N\cdot\mz)\bigg]^T.
\end{equation}
With this, imaging functional $\mathcal{I}_{\mathrm{LS}}(\mz;k)$ is defined as
\begin{equation}\label{ImagingLS}
  \mathcal{I}_{\mathrm{LS}}(\mz;k)=\left(\sum_{n=1}^{N}\frac{|\langle\mE_n(k),\mg_{\mathrm{D}}(\mz;k)\rangle|^2}{\sqrt{|\sigma_n(k)|}}\right)^{-1},
\end{equation}
where $\langle\mf,\mg\rangle=\mf\cdot\overline{\mg}$. Then the value of $\mathcal{I}_{\mathrm{LS}}(\mz;k)$ will be almost zero whenever $\mz\notin\Gamma$ and nonzero whenever $\mz=\my_m\in\Gamma$, for $m=1,2,\cdots,M$.

For TE case, the MSR matrix $\mathbb{K}(k)$ can be decomposed as
\begin{equation}\label{MSRN}
\mathbb{K}(k)=\frac{1-i}{4}\sqrt{\frac{k}{\pi}}\int_{\Gamma}\mathbb{E}_{\mathrm{N}}(\hat{\mx},\my;k)\mathbb{F}_{\mathrm{N}}(\hat{\mx},\my;k)^Td\my,
\end{equation}
where $\mathbb{E}_{\mathrm{N}}(\hat{\mx},\my;k)$ is the illumination vector
\begin{align}
\begin{aligned}\label{VecEN}
\mathbb{E}_{\mathrm{N}}(\hat{\mx},\my;k)&=-\bigg[\hat{\mx}_1\cdot\vn(\my)\exp(-ik\hat{\mx}_1\cdot \my),\cdots,\hat{\mx}_N\cdot\vn(\my)\exp(-ik\hat{\mx}_N\cdot\my)\bigg]^T\bigg|_{\hat{\mx}_j=-\vt_j}\\
&=\bigg[\vt_1\cdot\vn(\my)\exp(ik\vt_1\cdot\my),\cdots,\vt_N\cdot\vn(\my)\exp(ik\vt_N\cdot\my)\bigg]^T
\end{aligned}
\end{align}
and where $\mathbb{F}_{\mathrm{N}}(\hat{\mx},\my;k)$ is the corresponding density vector
\begin{equation}
\mathbb{F}_{\mathrm{N}}(\hat{\mx},\my;k)=\bigg[\psi(\my,\vt_1;k),\psi(\my,\vt_2;k),\cdots,\psi(\my,\vt_N;k)\bigg]^T.
\end{equation}

Formula (\ref{MSRN}) is a factorization of the MSR matrix that, like with the Dirichlet boundary condition case, separates the known incoming plane wave information from the unknown information. The range of $\mathbb{K}(k)$ is determined by the span of the $\mathbb{E}_{\mathrm{N}}(\hat{\mx},\my;k)$ corresponding to the $\Gamma$, i.e., we can define a signal subspace by using a set of left singular vectors of $\mathbb{K}(k)$.

The imaging algorithm for the Neumann boundary condition case is very similar to the Dirichlet boundary condition case. Based on the structure of (\ref{VecEN}), define a vector
\begin{equation}\label{testvector2}
\mg_{\mathrm{N}}(\mz;k)=\frac{1}{\sqrt{N}}\bigg[\vt_1\cdot\vn(\mz)\exp(ik\vt_1\cdot\mz),\cdots,\vt_N\cdot\vn(\mz)\exp(ik\vt_N\cdot\mz)\bigg]^T
\end{equation}
Since the unit normal $\vn(\mx)$ is unknown, for each point $\mz$ of the search domain, we use a set of directions $\vn_l(\mz)$ for $l=1,2,\cdots,L$, and we choose $\vn_l(\mx)$ which is to maximize the imaging functional among these directions at $\mz$. With this considerations, $\mathcal{I}_{\mathrm{LS}}(\mz;k)$ is defined as
\[\mathcal{I}_{\mathrm{LS}}(\mz;k)=\left\{\sum_{n=1}^{N}\left(\max_{1\leq l\leq L}\frac{|\langle\mE_n(k),\mg_{\mathrm{N}}(\mz;k)\rangle|^2}{\sqrt{|\sigma_n(k)|}}\right)\right\}^{-1}.\]

\begin{Rem}
  Due to the unknown of normal direction to $\Gamma$, mapping of $\mathcal{I}_{\mathrm{LS}}(\mz;k)$ requires large computational costs. Moreover, based on the results in \cite{Ppreprint,PL1,PL3}, the results are poor in general. Hence, throughout this paper, we apply (\ref{testvector}) for TE case instead of using (\ref{testvector2}) and explore its structure.
\end{Rem}

\section{Structure of imaging functional}\label{sec:3}
In this section, we carefully identify the structure of (\ref{ImagingLS}) for TM and TE cases. For this, we recall some useful results as follows.

\begin{lem}[See \cite{G,Ppreprint}]\label{TheoremBessel}
  For sufficiently large $N$, $\vt_n\in\mathbb{S}^1$, $n=1,2,\cdots,N$, and any two-dimensional vectors $\vx,\mx\in\mathbb{R}^2$, following relation holds:
  \begin{align*}
    &\sum_{n=1}^{N}\exp(i\omega\vt_n\cdot\mx)\approx\int_{\mathbb{S}^1}\exp(i\omega\vt\cdot\mx)d\vt=2\pi J_0(\omega|\mx|),\\
    &\sum_{n=1}^{N}\vt_n\cdot\vx\exp(ik\vt_n\cdot\mx)\approx\int_{\mathbb{S}^1}\vt\cdot\vx\exp(ik\vt\cdot\mx)d\vt=2\pi i\left(\frac{\mx}{|\mx|}\cdot\vx\right)J_1(k|\mx|),
  \end{align*}
  where $J_n$ denotes the Bessel function of integer order $n$ of the first kind.
\end{lem}

\begin{lem}[See \cite{PMUSIC}]
  Let $\mg_{\mathrm{D}}(\mz;k)$ be a test vector defined in (\ref{testvector}), $N$ is sufficiently large enough, and the Singular Value Decomposition (SVD) of $\mathbb{A}(k)$ is written by
  \[\mathbb{A}(k)=\mathbb{K}(k)^*\mathbb{K}(k)=\sum_{n=1}^{N}\tau_n(k)\mU_n(k)\mV_n(k)^*,\]
  where $\tau_n(k)$ denotes the singular values of $\mathbb{A}(k)$ satisfying
  \[\tau_1(k)\geq\tau_2(k)\geq\cdots\geq\tau_M(k)\]
  and
  \begin{equation}\label{smallSV}
    \tau_n(k)\approx0\quad\mbox{for}\quad n>M,
  \end{equation}
  $\mU_n(k)$ and $\mV_n(k)$ are left- and right-singular vectors of $\mathbb{A}$, respectively. Let $\mathbb{I}_N$ be the $N\times N$ identity matrix and
  \[\Psi(\mz,\my_m;k):=\left(\mathbb{I}_{N}-\sum_{n=1}^{M}\mU_n(k)\mU_n(k)^*\right)\mg_{\mathrm{D}}(\mz;k)=\left(\sum_{n=M+1}^{N}\mU_n(k)\mU_n(k)^*\right)\mg_{\mathrm{D}}(\mz;k).\]
  Then for TM case, $\Psi(\mz,\my_m;k)$ is represented as
  \begin{equation}\label{structureMUSIC1}
    \Psi(\mz,\my_m;k)=1-\sum_{m=1}^{M}J_0(k|\my_m-\mz|)^2.
  \end{equation}
  And for TE case, $\Psi(\mz,\my_n;k)$ is written as
  \begin{equation}\label{structureMUSIC2}
    \Psi(\mz,\my_m;k)=1-\sum_{m=1}^{M}\left(\frac{\my_m-\mz}{|\my_m-\mz|}\cdot\vn(\my_m)\right)^2J_1(k|\my_m-\mz|)^2.
  \end{equation}
\end{lem}

Now, we state the main result about the structure of imaging functional (\ref{ImagingLS}).
\begin{thm}\label{TheoremLS}
  Assume that $k$ and $N$ are sufficiently large enough and $\mg_{\mathrm{D}}(\mz;k)$ is defined in (\ref{testvector}). Let $\mathcal{N}:=\set{1,2,\cdots,N}$ and $\rho:\mathcal{N}\longrightarrow\mathcal{N}$ be a bijection. Then
  \begin{enumerate}
    \item For TM case, $\mathcal{I}_{\mathrm{LS}}(\mz;k)$ can be represented as follows.
      \begin{equation}\label{structureLS1}
        \mathcal{I}_{\mathrm{LS}}(\mz;k)\approx\left\{\sum_{n=1}^{M}\frac{J_0(k|\my_n-\mz|)^2}{\sqrt{|\sigma_{\rho^{-1}(n)}(k)|}} +\frac{1}{\sqrt{\epsilon}}\left(1-\sum_{n=1}^{M}J_0(k|\my_n-\mz|)^2\right)\right\}^{-1},
      \end{equation}
    \item For TE case, $\mathcal{I}_{\mathrm{LS}}(\mz;k)$ can be represented as follows.
      \begin{align}
      \begin{aligned}\label{structureLS2}
        \mathcal{I}_{\mathrm{LS}}(\mz;k)\approx&\left\{\sum_{n=1}^{M}\frac{1}{\sqrt{|\sigma_{\rho^{-1}(n)}(k)|}}\left(\frac{\my_n-\mz}{|\my_n-\mz|}\cdot\vn(\my_n)\right)^2J_1(k|\my_n-\mz|)^2\right.\\
        &\left.+\frac{1}{\sqrt{\epsilon}}\left(1-\sum_{n=1}^{M}\left(\frac{\my_n-\mz}{|\my_n-\mz|}\cdot\vn(\my_n)\right)^2J_1(k|\my_n-\mz|)^2\right)\right\}^{-1},
      \end{aligned}
      \end{align}
  \end{enumerate}
  where $\epsilon$ is a small value close to zero.
\end{thm}
\begin{proof}
  Based on the Eigenvalue Decomposition (ED) and SVD of hermitian matrix $\mathbb{A}(k)$, we can observe that
  \begin{equation}\label{EDSVD}
    \mathbb{A}(k)=\sum_{n=1}^{N}\sigma_n(k)\mE_n(k)\mE_n(k)^T=\sum_{n=1}^{N}\tau_n(k)\mU_n(k)\mV_n(k)^*.
  \end{equation}
  Therefore, for $n=1,2,\cdots,N$, $|\sigma_n(k)|\approx\tau_{\rho(n)}(k)$, and $\mE_n(k)\approx\mU_{\rho(n)}(k)$ or $\mE_n(k)\approx-\mU_{\rho(n)}(k)$. This means that (\ref{ImagingLS}) can be rewritten as
  \[\mathcal{I}_{\mathrm{LS}}(\mz;k)\approx\left(\sum_{n=1}^{N}\frac{|\langle\mU_{\rho(n)}(k),\mg_{\mathrm{D}}(\mz;k)\rangle|^2}{\sqrt{|\tau_{\rho(n)}(k)|}}\right)^{-1} =\left(\sum_{n=1}^{N}\frac{|\langle\mU_n(k),\mg_{\mathrm{D}}(\mz;k)\rangle|^2}{\sqrt{\tau_n(k)}}\right)^{-1}.\]

  Throughout this proof, based on (\ref{smallSV}), we let $\tau_n(k)\approx\epsilon$ for $n=M+1,M+2,\cdots,N$.

  \begin{enumerate}
    \item For TM case, Based on the physical factorization of MSR matrix, $\mU_n(k)$ is of the form (see \cite{HSZ} for instance)
      \[\mU_n(k)\approx\frac{1}{\sqrt{N}}\bigg[\exp(ik\vt_1\cdot\my_n),\exp(ik\vt_2\cdot\my_n),\cdots,\exp(ik\vt_N\cdot\my_n)\bigg]^T=\mg_{\mathrm{D}}(\my_n;k),\]
      for $n=1,2,\cdots,M$. With this, we can easily evaluate
      \begin{align}
      \begin{aligned}\label{term1}
        \sum_{n=1}^{M}\frac{|\langle\mU_n(k),\mg_{\mathrm{D}}(\mz;k)\rangle|^2}{\sqrt{\tau_n(k)}}&=\sum_{n=1}^{M}\frac{|\langle\mg_{\mathrm{D}}(\my_n;k),\mg_{\mathrm{D}}(\mz;k)\rangle|^2}{\sqrt{\tau_n(k)}}\\
        &=\sum_{n=1}^{M}\frac{1}{\sqrt{\tau_n(k)}}\left(\frac{1}{N}\sum_{p=1}^{N}\exp(ik\vt_p\cdot(\my_n-\mz))\right)^2\\
        &=\sum_{n=1}^{M}\frac{J_0(k|\my_n-\mz|)^2}{\sqrt{\tau_n(k)}}=\sum_{n=1}^{M}\frac{J_0(k|\my_n-\mz|)^2}{\sqrt{|\sigma_{\rho^{-1}(n)}(k)|}}.
      \end{aligned}
      \end{align}

      And, by (\ref{structureMUSIC1}), we can evaluate following
      \begin{align}
      \begin{aligned}\label{term2}
        \sum_{n=M+1}^{N}\frac{|\langle\mU_n(k),\mg_{\mathrm{D}}(\mz;k)\rangle|^2}{\sqrt{\tau_n(k)}}
        &\approx\sum_{n=M+1}^{N}\frac{\mg_{\mathrm{D}}(\mz;k)^*\mU_n(k)\mU_n(k)^*\mg_{\mathrm{D}}(\mz;k)}{\sqrt{\epsilon}}\\
        &=\frac{1}{\sqrt{\epsilon}}\mg_{\mathrm{D}}(\mz;k)^*\left(\sum_{n=M+1}^{N}\mU_n(k)\mU_n(k)^*\right)\mg_{\mathrm{D}}(\mz;k)\\
        &=\frac{1}{\sqrt{\epsilon}}\mg_{\mathrm{D}}(\mz;k)^*\left(\mathbb{I}_{N}-\sum_{n=1}^{M}\mU_n(k)\mU_n(k)^*\right)\mg_{\mathrm{D}}(\mz;k)\\
        &=\frac{1}{\sqrt{\epsilon}}\left(1-\sum_{n=1}^{M}|\langle\mU_n(k),\mg_{\mathrm{D}}(\mz;k)\rangle|^2\right)\\
        &=\frac{1}{\sqrt{\epsilon}}\left(1-\sum_{n=1}^{M}J_0(k|\my_n-\mz|)^2\right).
      \end{aligned}
      \end{align}
      Hence, (\ref{structureLS1}) can be derived by combining (\ref{term1}) and (\ref{term2}).
    \item For TE case, Based on the physical factorization of MSR matrix, $\mU_n(k)$ is of the form (see \cite{HSZ} for instance)
      \[\mU_n(k)\approx\frac{1}{\sqrt{N}}\bigg[\vt_1\cdot\vn(\my_n)\exp(ik\vt_1\cdot\my_n),\cdots,\vt_N\cdot\vn(\my_n)\exp(ik\vt_N\cdot\my_n)\bigg]^T=\mg_{\mathrm{N}}(\my_n;k),\]
      for $n=1,2,\cdots,M$. With this, we can evaluate
      \begin{align}
      \begin{aligned}\label{term3}
        \sum_{n=1}^{M}\frac{|\langle\mU_n(k),\mg_{\mathrm{D}}(\mz;k)\rangle|^2}{\sqrt{\tau_n(k)}}&=\sum_{n=1}^{M}\frac{|\langle\mg_{\mathrm{N}}(\my_n;k),\mg_{\mathrm{D}}(\mz;k)\rangle|^2}{\sqrt{\tau_n(k)}}\\
        &=\sum_{n=1}^{M}\frac{1}{\sqrt{\tau_n(k)}}\left(\frac{1}{N}\sum_{p=1}^{N}\vt_p\cdot\vn(\my_n)\exp(ik\vt_p\cdot(\my_n-\mz))\right)^2\\
        &=\sum_{n=1}^{M}\frac{1}{\sqrt{\tau_n(k)}}\left(\frac{\my_n-\mz}{|\my_n-\mz|}\cdot\vn(\my_n)\right)^2J_1(k|\my_n-\mz|)^2\\
        &=\sum_{n=1}^{M}\frac{1}{\sqrt{|\sigma_{\rho^{-1}(n)}(k)|}}\left(\frac{\my_n-\mz}{|\my_n-\mz|}\cdot\vn(\my_n)\right)^2J_1(k|\my_n-\mz|)^2.
      \end{aligned}
      \end{align}
      And, by (\ref{structureMUSIC2}), we can evaluate following
      \begin{align}
      \begin{aligned}\label{term4}
        \sum_{n=M+1}^{N}\frac{|\langle\mU_n(k),\mg_{\mathrm{D}}(\mz;k)\rangle|^2}{\sqrt{\tau_n(k)}}
        &\approx\sum_{n=M+1}^{N}\frac{\mg_{\mathrm{D}}(\mz;k)^*\mU_n(k)\mU_n(k)^*\mg_{\mathrm{D}}(\mz;k)}{\sqrt{\epsilon}}\\
        &=\frac{1}{\sqrt{\epsilon}}\mg_{\mathrm{D}}(\mz;k)^*\left(\sum_{n=M+1}^{N}\mU_n(k)\mU_n(k)^*\right)\mg_{\mathrm{D}}(\mz;k)\\
        &=\frac{1}{\sqrt{\epsilon}}\mg_{\mathrm{D}}(\mz;k)^*\left(\mathbb{I}_{N}-\sum_{n=1}^{M}\mU_n(k)\mU_n(k)^*\right)\mg_{\mathrm{D}}(\mz;k)\\
        &=\frac{1}{\sqrt{\epsilon}}\left(1-\sum_{n=1}^{M}|\langle\mU_n(k),\mg_{\mathrm{D}}(\mz;k)\rangle|^2\right)\\
        &=\frac{1}{\sqrt{\epsilon}}\left(1-\sum_{n=1}^{M}\left(\frac{\my_n-\mz}{|\my_n-\mz|}\cdot\vn(\my_n)\right)^2J_1(k|\my_n-\mz|)^2\right).
      \end{aligned}
      \end{align}
      Hence, (\ref{structureLS2}) can be derived by combining (\ref{term3}) and (\ref{term4}).
  \end{enumerate}
\end{proof}

Based on the recent works \cite{AGKPS,G,HHSZ,JKHP,JP,KP,P1,Ppreprint,P2,P3,P4,PL2,PP}, it is confirmed that applying multi-frequency offers better results than applying single-frequency in subspace migration and MUSIC algorithm. Hence, it is expected that following multi-frequency based imaging functional
\begin{equation}\label{ImagingLSM}
  \mathcal{I}_{\mathrm{LSM}}(\mz;k_1,k_{F})=\left(\sum_{f=1}^{F}\sum_{n=1}^{N}\frac{|\langle\mE_n(k_f),\mg_{\mathrm{D}}(\mz;k_f)\rangle|^2}{\sqrt{|\sigma_n(k_f)|}}\right)^{-1},
\end{equation}
should be an improved version of (\ref{ImagingLS}). Following result supports this fact. The proof is similar to the \cite[Theorem 3.2]{Ppreprint} and \cite[Theorem 3.9]{Ppreprint} for TM and TE case, respectively.

\begin{thm}
  Assume that $k_1<k_2<\cdots<k_F$, and $N$ are sufficiently large enough and $\mg_{\mathrm{D}}(\mz;k_f)$ is defined in (\ref{testvector}). Let $\mathcal{N}:=\set{1,2,\cdots,N}$, $\rho_f:\mathcal{N}\longrightarrow\mathcal{N}$ be a one-to-one function for each $f=1,2,\cdots,F$, and $\epsilon$ is a small value close to zero. Then
  \begin{enumerate}
    \item For TM case, $\mathcal{I}_{\mathrm{LSM}}(\mz;k_1,k_F)$ can be represented as follows
      \begin{equation}\label{structureLSM1}
        \mathcal{I}_{\mathrm{LSM}}(\mz;k_1,k_F)\approx\left\{\sum_{f=1}^{F}\sum_{n=1}^{M}\frac{J_0(k_f|\my_n-\mz|)^2}{\sqrt{|\sigma_{\rho_f^{-1}(n)}(k)|}} +\frac{F}{\sqrt{\epsilon}}\left(1-\sum_{n=1}^{M}\Lambda_1(\mz,\my_n;k_1,k_F)\right)\right\}^{-1},
      \end{equation}
      where
      \begin{multline*}
        \Lambda_1(\mz,\my_n;k_1,k_F)=\frac{k_F}{k_F-k_1}\bigg(J_0(k_F|\my_n-\mz|)^{2}+J_1(k_F|\my_n-\mz|)^{2}\bigg)\\
        -\frac{k_1}{k_F-k_1}\bigg(J_0(k_1|\my_n-\mz|)^{2}+J_1(k_1|\my_n-\mz|)^{2}\bigg).
      \end{multline*}
    \item For TE case, $\mathcal{I}_{\mathrm{LSM}}(\mz;k_1,k_F)$ can be represented as follows
      \begin{align}
      \begin{aligned}\label{structureLSM2}
        \mathcal{I}_{\mathrm{LSM}}(\mz;k_1,k_F)\approx&\left\{\sum_{f=1}^{F}\sum_{n=1}^{M}\frac{1}{\sqrt{|\sigma_{\rho_f^{-1}(n)}(k)|}}\left(\frac{\my_n-\mz}{|\my_n-\mz|}\cdot\vn(\my_n)\right)^2J_1(k_f|\my_n-\mz|)^2\right.\\
        &\left.+\frac{F}{\sqrt{\epsilon}}\left(1-\sum_{n=1}^{M}\left(\frac{\my_n-\mz}{|\my_n-\mz|}\cdot\vn(\my_n)\right)^2\int_{k_1}^{k_F}J_1(k|\my_n-\mz|)^2dk\right)\right\}^{-1}.
      \end{aligned}
      \end{align}
  \end{enumerate}
\end{thm}

\section{Numerical examples}\label{sec:4}
In this section, some numerical results are exhibited for supporting identified structure (\ref{structureLS1}). For this, three small cracks
\begin{align*}
  \Gamma_1&=\set{[s-0.6,-0.2]^\mathrm{T}:-\ell\leq s\leq\ell}\\
  \Gamma_2&=\set{R_{\pi/4}[s+0.4,s+0.35]^\mathrm{T}:-\ell\leq s\leq\ell}\\
  \Gamma_3&=\set{R_{7\pi/6}[s+0.25,s-0.6]^\mathrm{T}:-\ell\leq s\leq\ell}
\end{align*}
and three extended cracks are chosen:
\begin{align*}
  \Gamma_4&=\set{\left[s,\frac{1}{2}\cos\frac{s\pi}{2}+\frac{1}{5}\sin\frac{s\pi}{2}-\frac{1}{10}\cos\frac{3s\pi}{2}\right]^\mathrm{T}:s\in[-1,1]}\\
  \Gamma_5&=\set{\left[2\sin\frac{s}{2},\sin s\right]^\mathrm{T}:s\in\left[\frac{\pi}{4},\frac{7\pi}{4}\right]}\\
  \Gamma_6&=\set{\left[s-0.2,-0.5s^2+0.6\right]^\mathrm{T}:s\in[-0.5,0.5]}\cup\set{\left[s+0.2,s^3+s^2-0.6\right]^\mathrm{T}:s\in[-0.5,0.5]}.
\end{align*}
Here, $R_{\theta}$ denotes the rotation by $\theta$ and $\ell=0.05$. We apply $N=12$ different incident and observation directions for small ($\Gamma_1,\Gamma_2,\Gamma_3$), $N=32$ (TM case) and $N=48$ (TE case) directions for extended cracks ($\Gamma_4,\Gamma_5,\Gamma_6$), respectively, such that
\[\vt_n=\left[\cos\frac{2\pi n}{N},\sin\frac{2\pi n}{N}\right]^\mathrm{T},\]
and wavenumber $k_f$ is of the form $k=2\pi/\lambda_f$, where $\lambda_f$ is the given wavelength for $f=1,2,\cdots,F=10$.

It is worth emphasizing that every Far-field pattern datas $u_{\infty}(\hat{\mx},\vt;k)$ of (\ref{FFPD}) and (\ref{FFPN}) are generated by the formulation involving the solution of a second-kind Fredholm integral equation along the crack (see \cite[Chapter 3]{N} and \cite[Chapter 4]{N} for single and multiple cracks, respectively) in order to avoid \textit{inverse crime}. After obtaining the dataset, a $20$dB white Gaussian random noise is added.

Figure \ref{DistrributionEVSV} shows the distributions of eigenvalue and singular values of $\mathbb{A}(k)$ for $\lambda=0.4$ when the crack is $\Gamma_4$ in order to examine (\ref{EDSVD}). Based on this figure, we can easily observe that the values $|\sigma_n(k)|$, $n=1,2,\cdots,N$, are almost same as $\tau_{\rho(n)}(k)$ when $\rho(n)=N-n+1$. In Table \ref{Errors}, evaluated values of $|\mE_n(k)-\mU_{\rho(n)}(k)|$ are shown for $n=1,2,\cdots,20$. These values tells us that every eigenvectors $\mE_n(k)$ and singular vectors are $\mU_{\rho(n)}(k)$ almost same except the directions (for example, $\mE_1(k),\mE_3(k),\mE_5(k),\cdots$).

\begin{figure}[!ht]
\begin{center}
\includegraphics[width=0.49\textwidth]{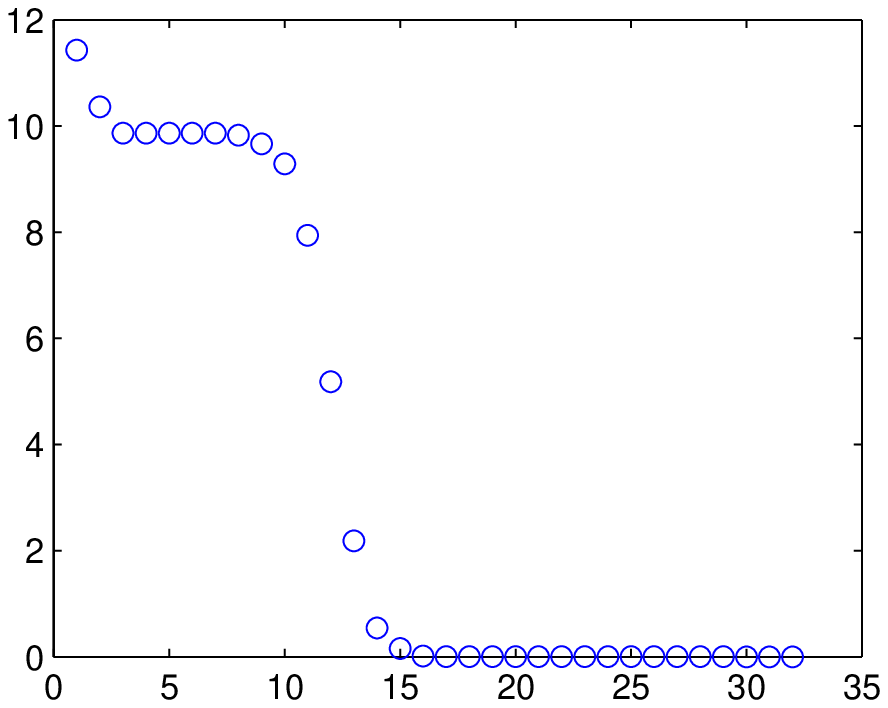}
\includegraphics[width=0.49\textwidth]{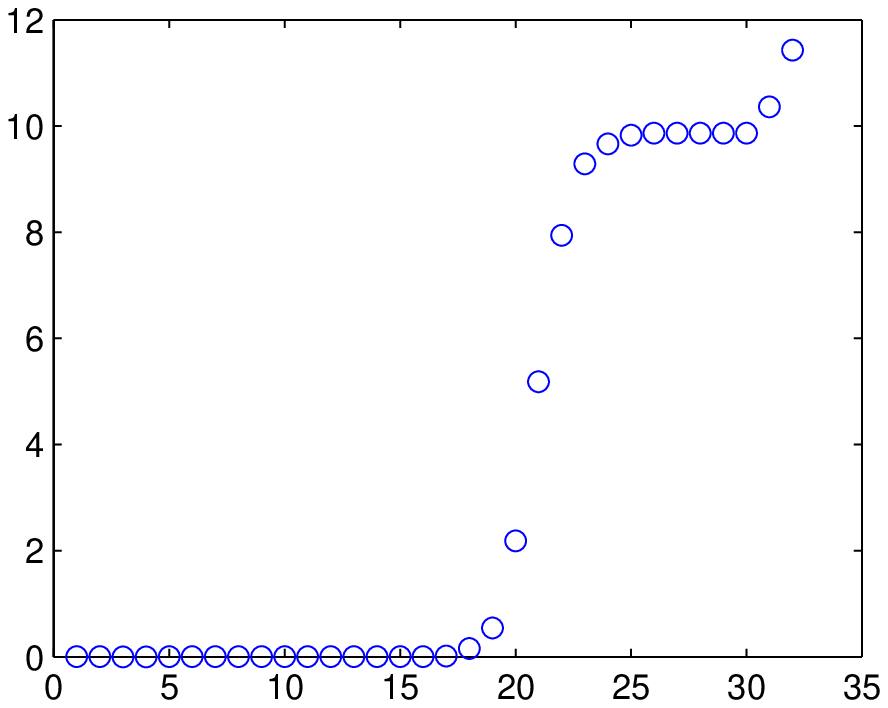}
\caption{\label{DistrributionEVSV}Distribution of singular values (left) and eigenvalues (right) of $\mathbb{A}(k)$ when $\lambda=0.4$. In this case, $\rho(n)=N-n+1$.}
\end{center}
\end{figure}

\begin{table}
\begin{center}
\begin{tabular}{c|c||c|c}
\hline
$n$&$|\mE_n(k)-\mU_{\rho(n)}(k)|$&$n$&$|\mE_n(k)-\mU_{\rho(n)}(k)|$\\
\hline\hline
$1$&$2.00000000000000000000$&$11$&$1.65562620502031\mathrm{e}-15$\\
$2$&$6.65870414549247\mathrm{e}-15$&$12$&$2.41720215436400\mathrm{e}-15$\\
$3$&$2.00000000000000000000$&$13$&$2.43989740405264\mathrm{e}-15$\\
$4$&$1.84821973386216\mathrm{e}-10$&$14$&$2.00000000000000000000$\\
$5$&$2.00000000000000000000$&$15$&$1.57744220582806\mathrm{e}-14$\\
$6$&$1.56714451336802\mathrm{e}-11$&$16$&$2.00000000000000000000$\\
$7$&$2.00000000000000000000$&$17$&$2.00000000000000000000$\\
$8$&$5.94301230666634\mathrm{e}-14$&$18$&$9.46987310586183\mathrm{e}-12$\\
$9$&$2.00000000000000000000$&$19$&$2.00000000000000000000$\\
$10$&$2.00000000000000000000$&$20$&$6.20211785729028\mathrm{e}-10$\\
\hline
\end{tabular}
\caption{\label{Errors}Evaluated values of $|\mE_n(k)-\mU_{\rho(n)}(k)|$ with $\rho(n)=N-n+1$.}
\end{center}
\end{table}

Figure \ref{SmallCracks} shows the maps of $\mathcal{I}_{\mathrm{LS}}(\mz;k_{10})$ and $\mathcal{I}_{\mathrm{LSM}}(\mz;k_1,k_{10})$ (right) for $\lambda_1=0.6$ and $\lambda_{10}=0.4$. As we observed in Theorem \ref{TheoremLS}, $\mathcal{I}_{\mathrm{LS}}(\mz;k_{10})$ plots peak of large magnitude at the location of $\Gamma_j$, small magnitude at $\mz\notin\Gamma_j$ for $j=1,2,3$, and unexpected artifacts. However, in the map of $\mathcal{I}_{\mathrm{LSM}}(\mz;k_1,k_{10})$, there artifacts are successfully eliminated hence, we can identify (locations and rotations) cracks more easier.

\begin{figure}[!ht]
\begin{center}
\includegraphics[width=0.49\textwidth]{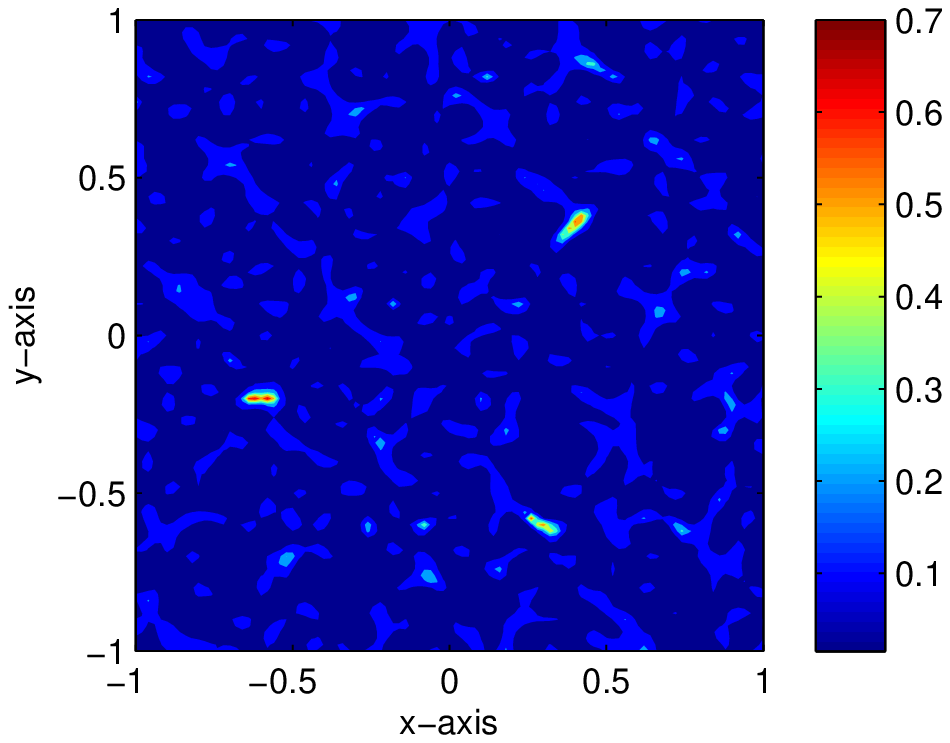}
\includegraphics[width=0.49\textwidth]{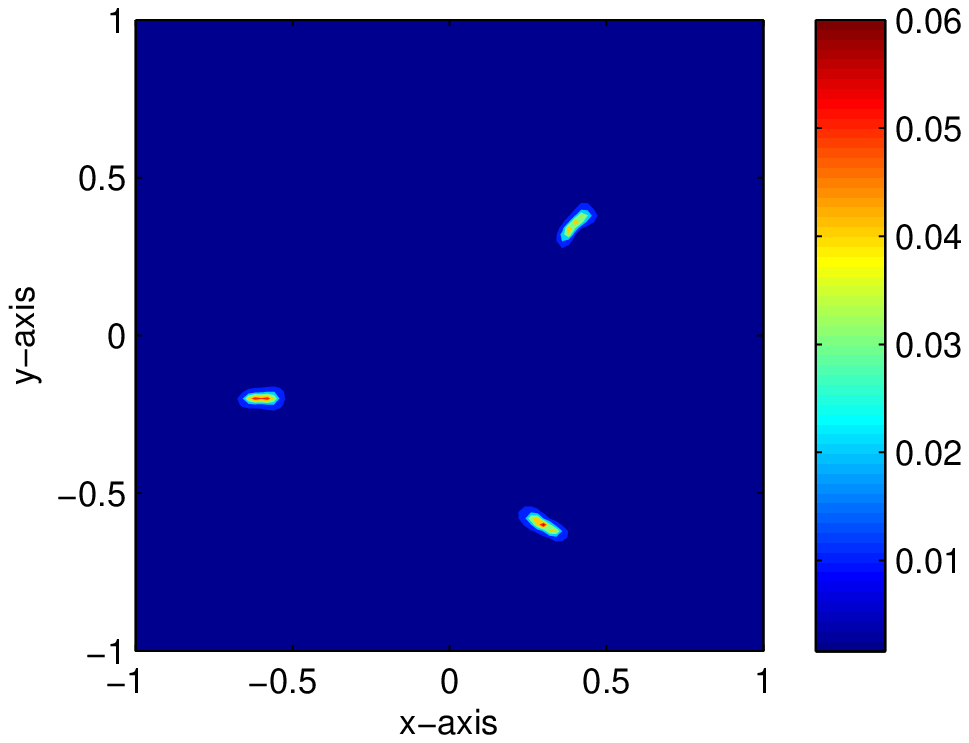}
\caption{\label{SmallCracks} (TM case) Maps of $\mathcal{I}_{\mathrm{LS}}(\mz;k_{10})$ (left) and $\mathcal{I}_{\mathrm{LSM}}(\mz;k_1,k_{10})$ (right) for $\lambda_1=0.6$ and $\lambda_{10}=0.4$ when the cracks are $\Gamma_1\cup\Gamma_2\cup\Gamma_3$.}
\end{center}
\end{figure}

Figures \ref{SingleCrackD1}, \ref{SingleCrackD2}, and \ref{SingleCrackD3} show the maps of $\mathcal{I}_{\mathrm{LS}}(\mz;k_{10})$ and $\mathcal{I}_{\mathrm{LSM}}(\mz;k_1,k_{10})$ (right) for $\lambda_1=0.7$ and $\lambda_{10}=0.4$ in TM case. Similar to the imaging of small cracks, we can identify that single-frequency linear sampling method offers very good result but unexpected artifacts are still remaining. However, multi-frequency linear sampling method successfully eliminates so that we can identify the shape of extended cracks more accurately.

\begin{figure}[!ht]
\begin{center}
\includegraphics[width=0.49\textwidth]{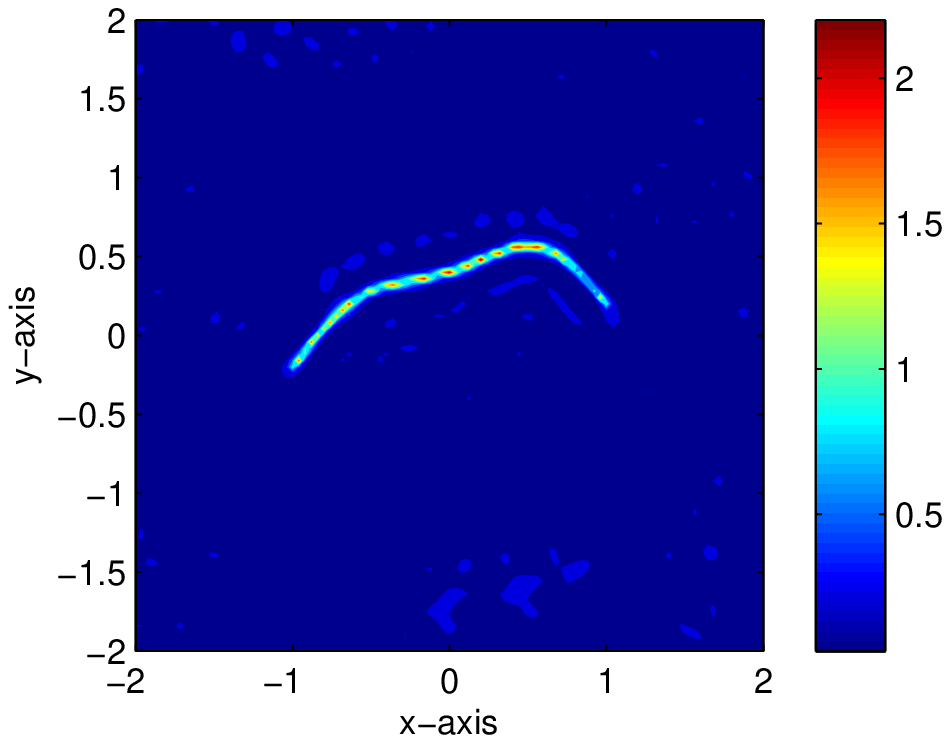}
\includegraphics[width=0.49\textwidth]{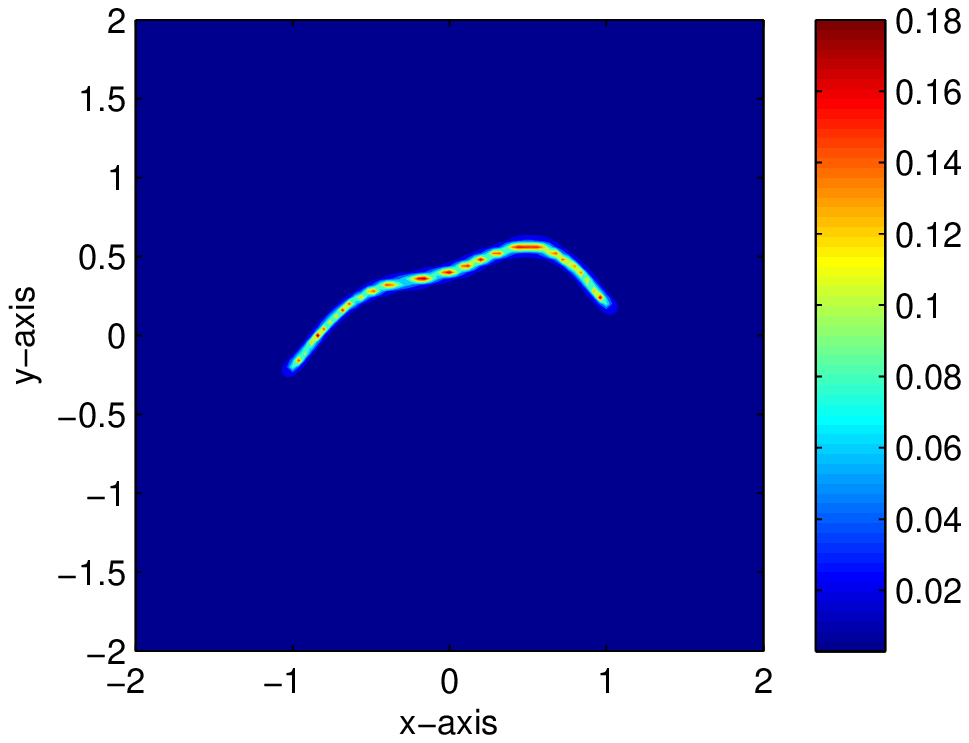}
\caption{\label{SingleCrackD1} (TM case) Maps of $\mathcal{I}_{\mathrm{LS}}(\mz;k_{10})$ (left) and $\mathcal{I}_{\mathrm{LSM}}(\mz;k_1,k_{10})$ (right) for $\lambda_1=0.7$ and $\lambda_{10}=0.4$ when the crack is $\Gamma_4$.}
\end{center}
\end{figure}

\begin{figure}[!ht]
\begin{center}
\includegraphics[width=0.49\textwidth]{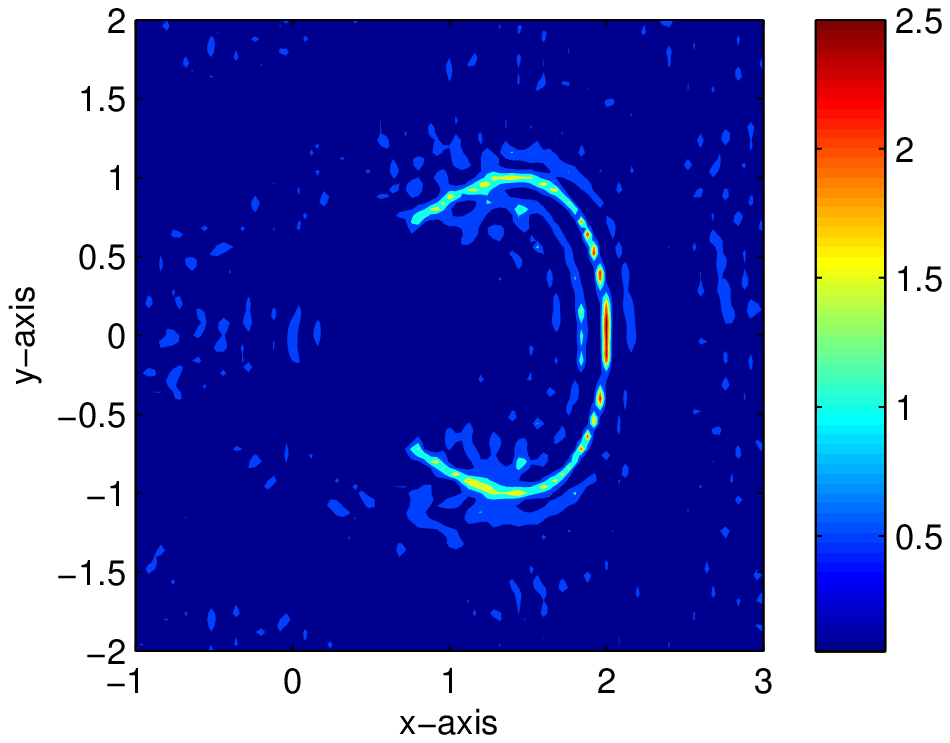}
\includegraphics[width=0.49\textwidth]{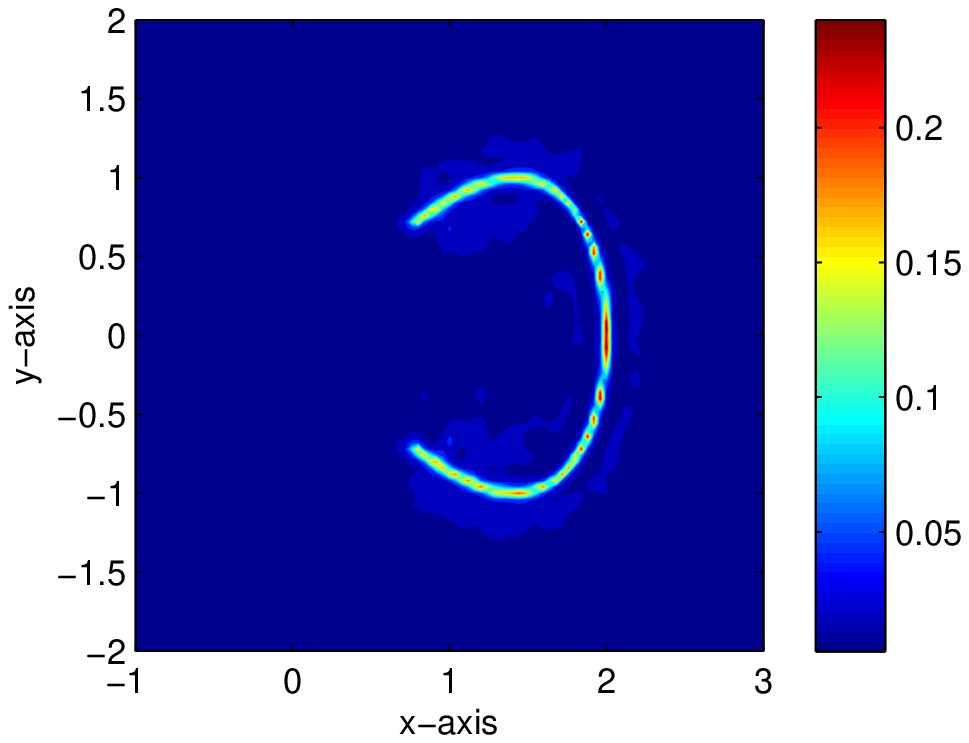}
\caption{\label{SingleCrackD2} (TM case) Same as Figure \ref{SingleCrackD1} except the crack is $\Gamma_5$.}
\end{center}
\end{figure}

\begin{figure}[!ht]
\begin{center}
\includegraphics[width=0.49\textwidth]{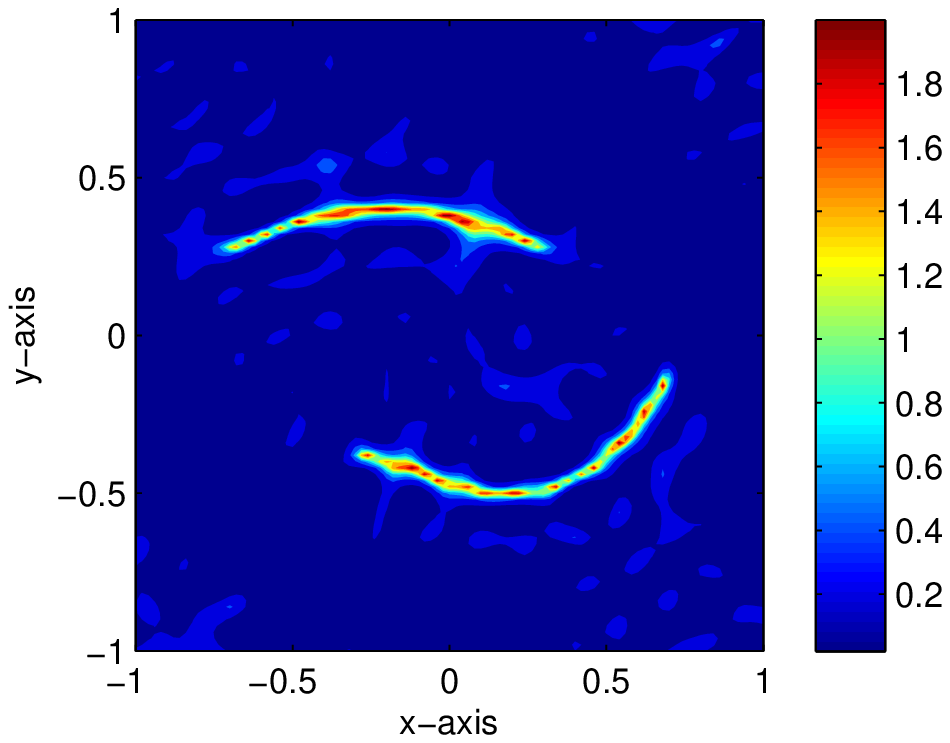}
\includegraphics[width=0.49\textwidth]{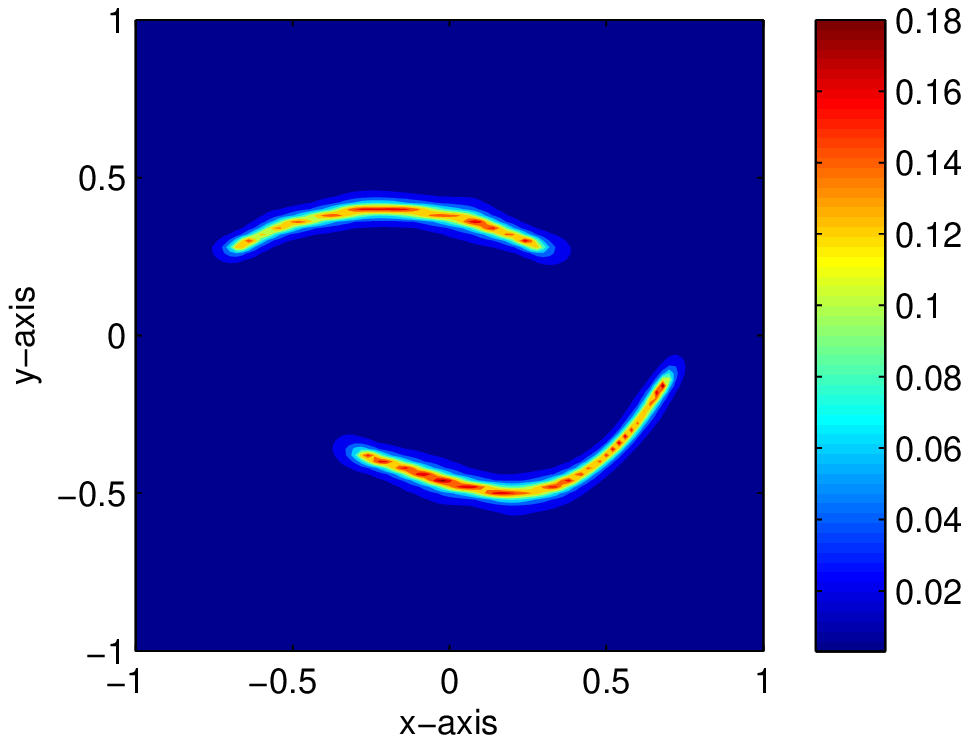}
\caption{\label{SingleCrackD3} (TM case) Same as Figure \ref{SingleCrackD1} except the cracks are $\Gamma_6$.}
\end{center}
\end{figure}

Figures \ref{SingleCrackN1}, \ref{SingleCrackN2}, and \ref{SingleCrackN3} show the maps of $\mathcal{I}_{\mathrm{LS}}(\mz;k_{10})$ and $\mathcal{I}_{\mathrm{LSM}}(\mz;k_1,k_{10})$ (right) for $\lambda_1=0.6$ and $\lambda_{10}=0.4$ in TE case. As we observed in Theorem \ref{TheoremLS}, $\mathcal{I}_{\mathrm{LS}}(\mz;k_{10})$ and $\mathcal{I}_{\mathrm{LSM}}(\mz;k_1,k_{10})$ generate two curves with large magnitude in the neighborhood of cracks. Opposite to the TM case, $\mathcal{I}_{\mathrm{LS}}(\mz;k_{10})$ does not produces good result but $\mathcal{I}_{\mathrm{LSM}}(\mz;k_1,k_{10})$ still yields acceptable results.

\begin{figure}[!ht]
\begin{center}
\includegraphics[width=0.49\textwidth]{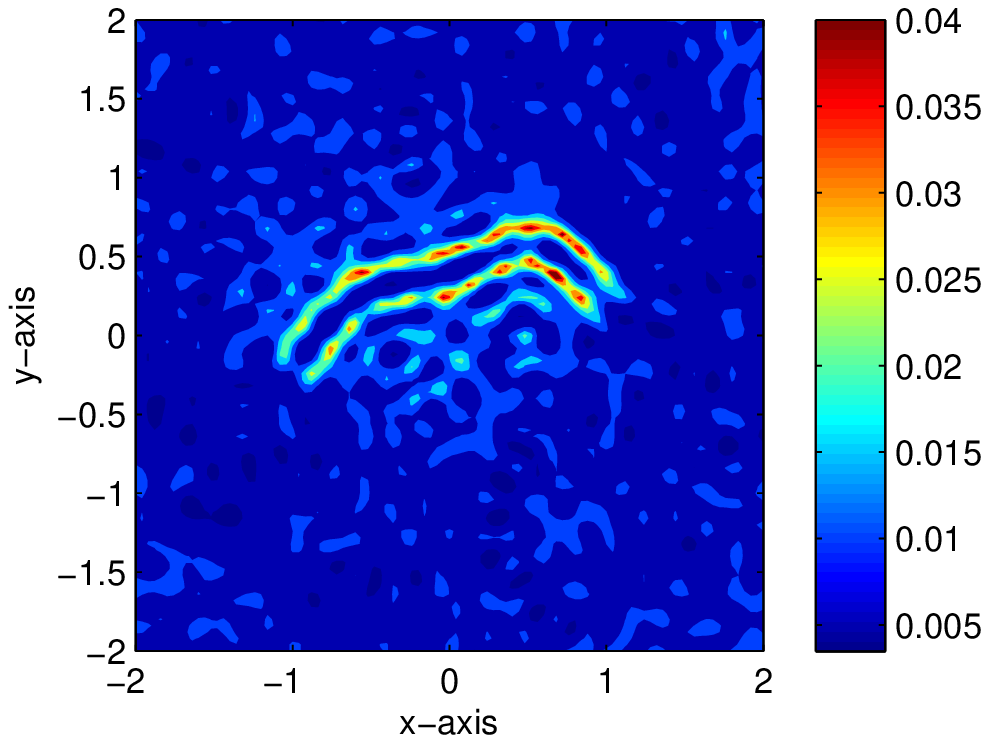}
\includegraphics[width=0.49\textwidth]{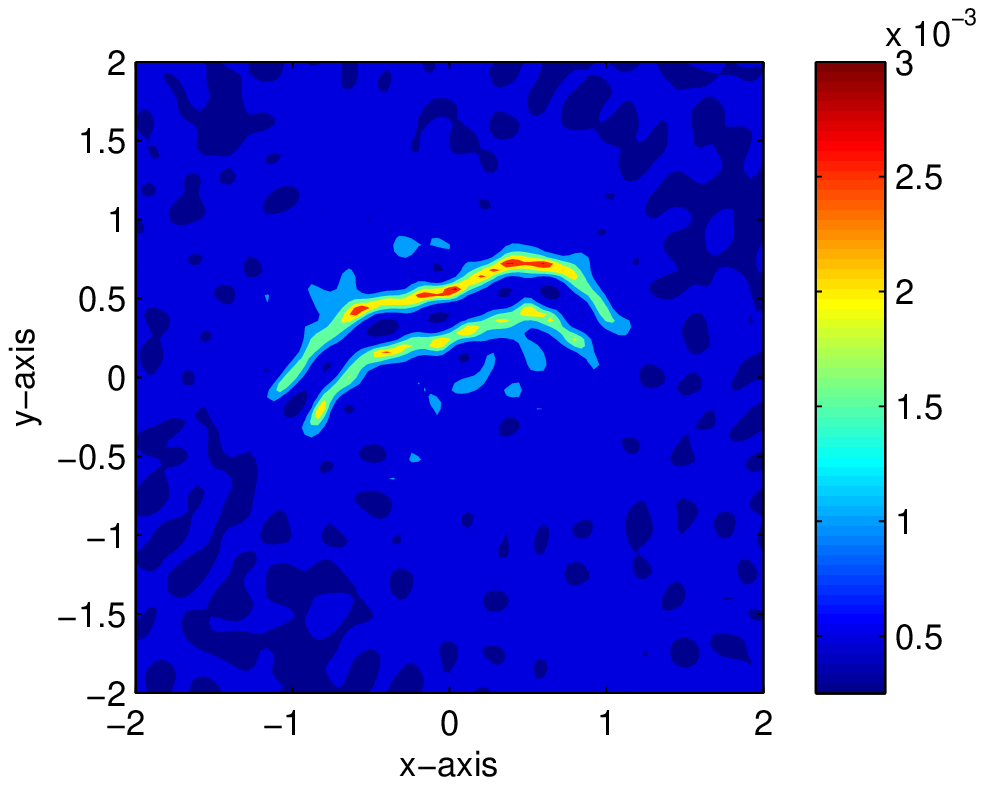}
\caption{\label{SingleCrackN1} (TE case) Maps of $\mathcal{I}_{\mathrm{LS}}(\mz;k_{10})$ (left) and $\mathcal{I}_{\mathrm{LSM}}(\mz;k_1,k_{10})$ (right) for $\lambda_1=0.7$ and $\lambda_{10}=0.4$ when the crack is $\Gamma_4$.}
\end{center}
\end{figure}

\begin{figure}[!ht]
\begin{center}
\includegraphics[width=0.49\textwidth]{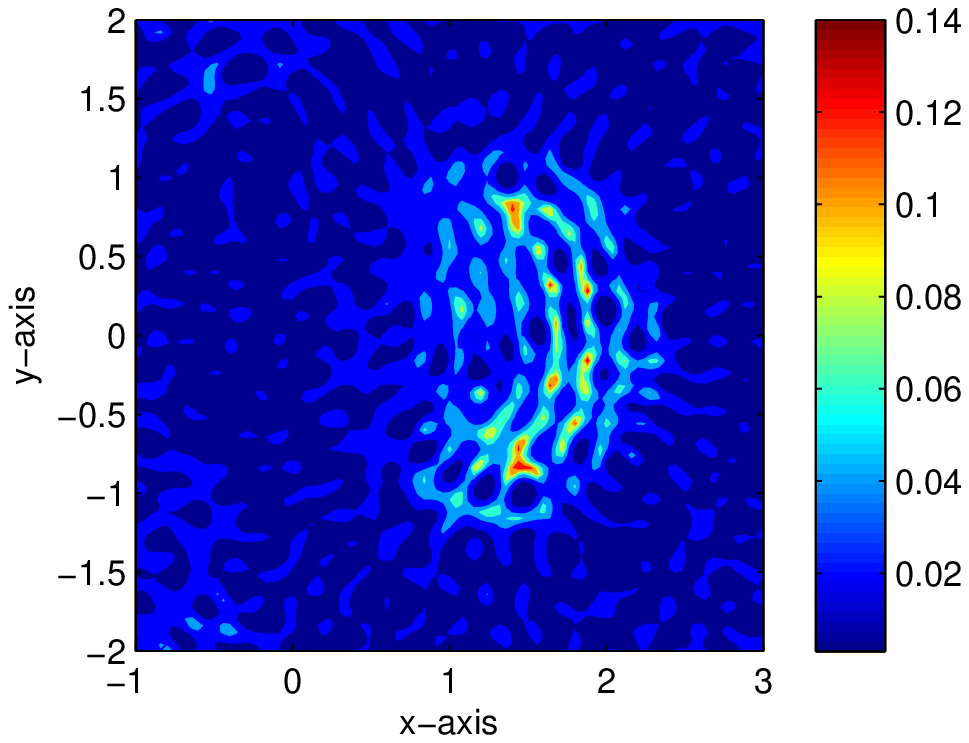}
\includegraphics[width=0.49\textwidth]{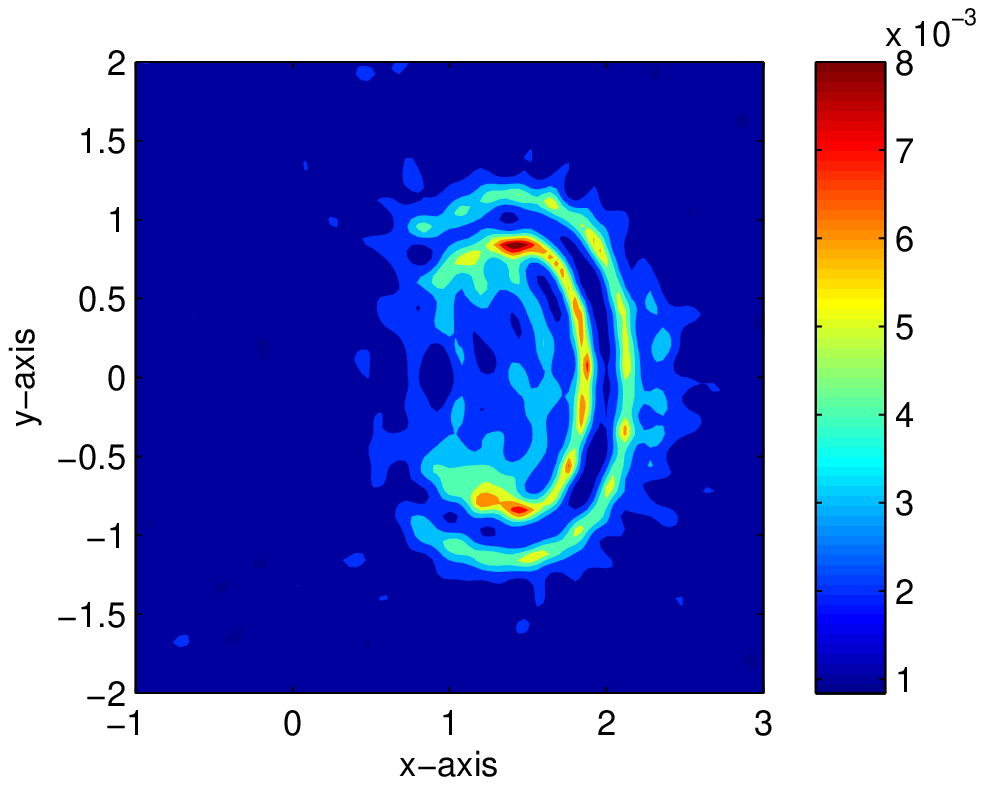}
\caption{\label{SingleCrackN2} (TE case) Same as Figure \ref{SingleCrackN1} except the crack is $\Gamma_5$.}
\end{center}
\end{figure}

\begin{figure}[!ht]
\begin{center}
\includegraphics[width=0.49\textwidth]{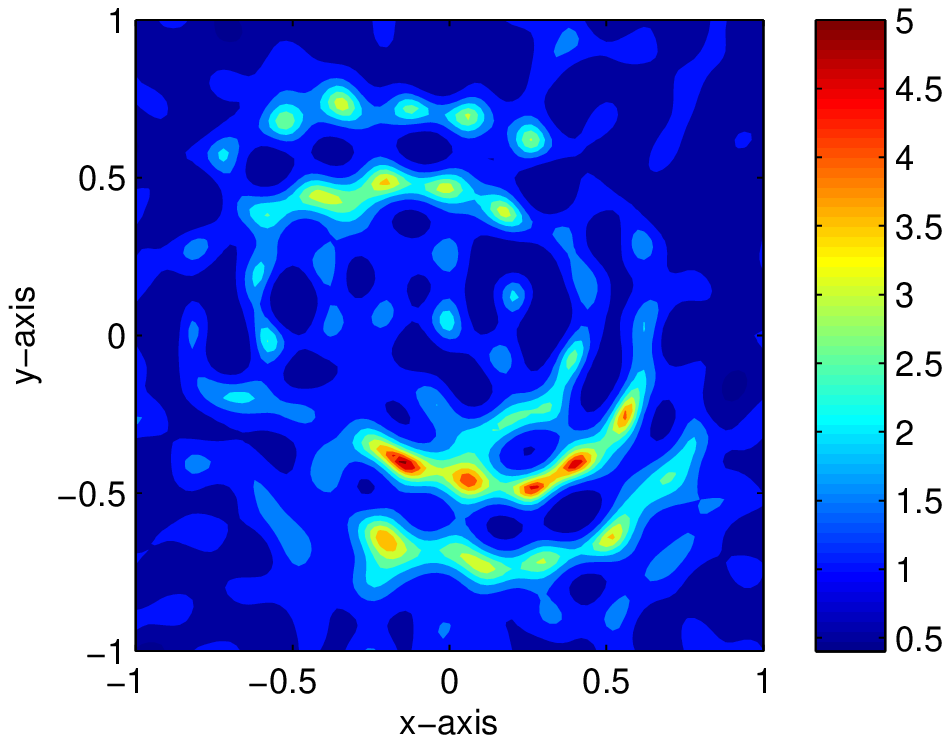}
\includegraphics[width=0.49\textwidth]{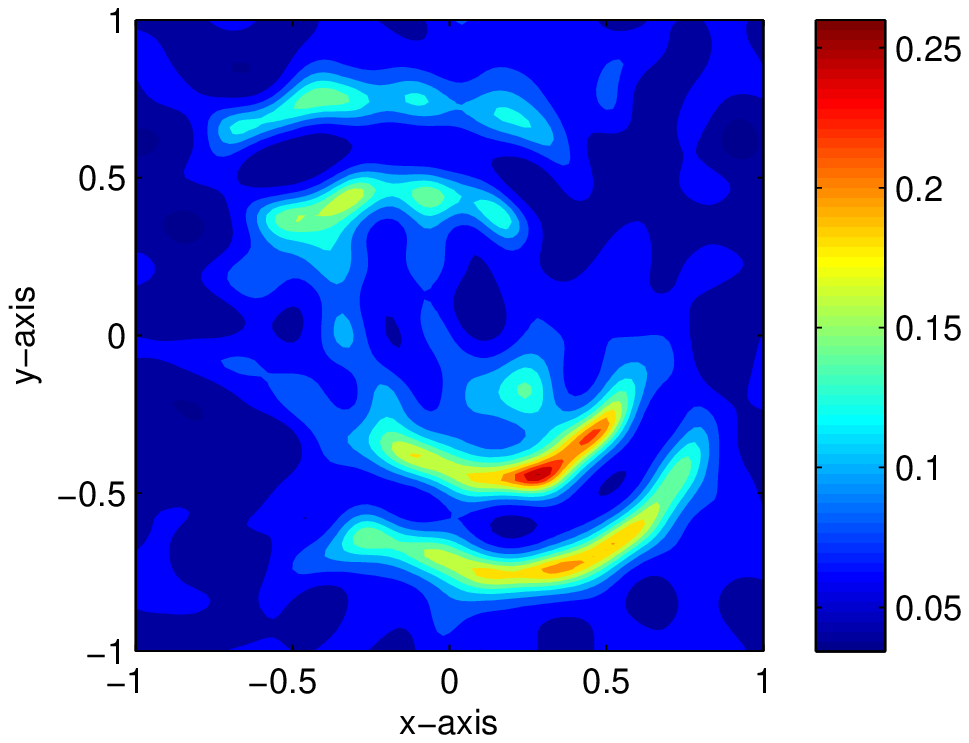}
\caption{\label{SingleCrackN3} (TE case) Same as Figure \ref{SingleCrackN1} except the cracks are $\Gamma_6$.}
\end{center}
\end{figure}

\section{Concluding remark}\label{sec:5}
Based on the structure and relationship of eigenvectors and singular vectors of MSR matrix, and integral representation formula of the Bessel function, we examined the structure of single- and multi-frequency electromagnetic imaging functions adopted in the famous linear sampling method. Because of the oscillation aspect of the Bessel function, we confirmed the reason behind the improved imaging performance by successfully applying high and multiple frequencies.

Throughout numerical results, it is confirmed that linear sampling method offers good results but they do not guarantee complete shaping of cracks. However, they can be adopted as a good initial guess of a level-set evolution or of a standard iterative algorithm \cite{ADIM,DL,K,M2,PL4} so that more accurate shape of cracks can be obtained.

Based on recent works in \cite{AGKPS,JKP,KP}, it has shown that subspace migration and MUSIC algorithm can be applied to limited-view inverse scattering problems for imaging of small targets. Motivated this, identifying the structure of imaging function used in the linear sampling method in the limited-view problem should be an interesting research topic. Moreover, discovering some properties of linear sampling method in the imaging of perfectly conducting cracks with Neumann boundary condition will be a remarkable subject.

In this paper, we considered the linear sampling method for imaging of perfectly conducting cracks. Extension to the arbitrary shaped extended target and to the three-dimensional problem will be an interesting problem.

\end{document}